\documentclass[conference,letterpaper]{IEEEtran}
\IEEEoverridecommandlockouts
\usepackage{cite}
\usepackage{amsmath}
\usepackage{amsfonts}
\usepackage{amssymb}
\usepackage{algorithm}
\usepackage{algorithmic}
\usepackage{graphicx}
\usepackage{verbatim}
\usepackage{amsthm}
\usepackage{xcolor}
\usepackage{multirow}

\newtheorem{theorem}{Theorem}
\newtheorem{lemma}{Lemma}
\newtheorem{remark}{Remark}

\newtheorem{example}{Example}

\newcommand{\beq}{\begin{equation}}
\newcommand{\eeq}{\end{equation}}
\newcommand{\beqnn}{\begin{equation*}}
\newcommand{\eeqnn}{\end{equation*}}
\newcommand{\beqy}{\begin{eqnarray}}
\newcommand{\eeqy}{\end{eqnarray}}
\newcommand{\beqynn}{\begin{eqnarray*}}
\newcommand{\eeqynn}{\end{eqnarray*}}
\newcommand{\bit}{\begin{itemize}}
\newcommand{\eit}{\end{itemize}}
\newcommand{\ben}{\begin{enumerate}}
\newcommand{\een}{\end{enumerate}}
\newcommand{\bex}{\begin{example}}
\newcommand{\eex}{\end{example}}

\newcommand{\balg}[1]{\begin{algorithm} \caption{#1}}
\newcommand{\ealg}{\end{algorithm}}

\newcommand{\balgc}{\begin{algorithmic}[1]}
\newcommand{\ealgc}{\end{algorithmic}}

\newcommand{\bary}{\begin{array}}
\newcommand{\eary}{\end{array}}
\newcommand{\bmx}{\begin{bmatrix}}
\newcommand{\emx}{\end{bmatrix}}
\newcommand{\bsmx}{\left[\begin{smallmatrix}}
\newcommand{\esmx}{\end{smallmatrix}\right]}
\newcommand{\bmxc}[1]{\left[\begin{array}{@{}#1@{}}}
\newcommand{\emxc}{\end{array}\right]}
\newcommand{\bcn}{\begin{center}}
\newcommand{\ecn}{\end{center}}

\newcommand{\Rbb}{{\mathbb{R}}}
\newcommand{\Zbb}{{\mathbb{Z}}}

\newcommand{\bigO}{{\mathcal{O}}}

\newcommand{\Rnbn}{\Rbb^{n \times n}}

\newcommand{\Rmbn}{\Rbb^{m \times n}}

\newcommand{\Zn}{\Zbb^{n}}

\newcommand{\A}{\boldsymbol{A}}

\newcommand{\Q}{\boldsymbol{Q}}
\newcommand{\R}{\boldsymbol{R}}

\newcommand{\Z}{\boldsymbol{Z}}

\newcommand{\x}{{\boldsymbol{x}}}
\newcommand{\y}{{\boldsymbol{y}}}

\newcommand{\0}{{\boldsymbol{0}}}

\usepackage{slashbox}

\allowdisplaybreaks

\begin{document}

\title{Improved Upper Bounds \\ on the Hermite and KZ Constants}

\author{\IEEEauthorblockN{Jinming~Wen\IEEEauthorrefmark{1} \IEEEauthorrefmark{2},
Xiao-Wen~Chang\IEEEauthorrefmark{3}  and Jian Weng\IEEEauthorrefmark{1}}
\IEEEauthorblockA{\IEEEauthorrefmark{1}
College of Information Science and Technology and the College of Cyber Security, Jinan University, \\
Guangzhou, 510632, China  (E-mail:jinming.wen@mail.mcgill.ca, cryptjweng@gmail.com)}% This "%" stops a space
\IEEEauthorblockA{\IEEEauthorrefmark{2}State Key Laboratory of Information Security, Institute of Information Engineering,\\ Chinese Academy of Sciences, Beijing 100093}
\IEEEauthorblockA{\IEEEauthorrefmark{3}School of Computer Science, McGill University,
Montreal, H3A 0E9, Canada (E-mail: chang@cs.mcgill.ca)}

\thanks{This work was partially supported by  NSERC of Canada grant RGPIN-2017-05138,
NNSFC (No. 11871248),
and the Fundamental Research Funds for the Central Universities (No. 21618329).}}

\maketitle

\begin{abstract}
The Korkine-Zolotareff (KZ) reduction is a widely used lattice reduction strategy  in communications and cryptography.
The Hermite constant, which is a vital constant of lattice, has many applications,
such as bounding the length of the shortest nonzero lattice vector and orthogonality defect of lattices.
The KZ constant can be used in quantifying some useful properties of KZ reduced matrices.
In this paper, we first develop a linear upper bound on the Hermite constant
and then use the bound to develop an upper bound on the KZ constant.
These upper bounds are sharper than those obtained recently by the first two authors.
Some examples on the applications of the improved upper bounds are also presented.
\end{abstract}

\begin{IEEEkeywords}
KZ reduction,  Hermite constant, KZ constant.
\end{IEEEkeywords}


% \fi
\IEEEpeerreviewmaketitle

\section{Introduction}
\label{s:introduction}
The lattice  generated by a matrix $\A\in \mathbb{R}^{m\times n}$ with full-column rank
is defined by
\beq
\label{e:latticeA}
\mathcal{L}(\A)=\{\A\x \,|\,\x \in \mathbb{Z}^n\}.
\eeq
The column vectors of $\A$ and $n$ represent the basis and dimension of $\mathcal{L}(\A)$, respectively.

A matrix $\Z \in \mathbb{Z}^{n\times n}$ satisfying $|\det(\Z)|=1$ is said to be unimodular.
For any unimodular $\Z \in \mathbb{Z}^{n\times n}$, $\mathcal{L}(\A\Z)$ is the same lattice as $\mathcal{L}(\A)$.
Lattice reduction is the process of finding a unimodular $\Z$ such that the column vectors of
$\A\Z$ are short.
There are a few types of lattice reduction strategies.
The Lenstra-Lenstra-Lov\'asz (LLL) reduction and the Korkine-Zolotareff (KZ) reduction are two of the most popular ones,
and they have crucial applications in many domains including communications \cite{AgrEVZ02} and cryptography \cite{MicR08}.

For efficiency, the LLL reduction is often used to preprocess the matrix $\A$ when a closest vector problem (CVP), which is defined as
\beqnn
\min_{\x\in\mathbb{Z}^n}\|\y-\A\x\|_2,
\eeqnn
needs to be solved.
In some communications applications,  a number of  CVPs
with the same matrix $\A$ but different $\y$ need to be solved.
In this situation, for efficiency, instead of the LLL reduction, the KZ reduction is applied to preprocess $\A$.
The reason is that although it is more time consuming to perform the KZ reduction than the LLL reduction,
the reduced matrix of the KZ reduction has better properties than the one obtained by the LLL reduction,
and hence the total computational time of solving these CVPs by using the KZ reduction may be less than that of
using the LLL reduction.
Furthermore, the KZ reduction finds applications in successive integer-forcing linear receiver design \cite{OrdEN13}
and integer-forcing linear receiver design \cite{SakHV13}.

It is interesting to quantify the performance of the KZ reduction in terms of shortening the lengths of the lattice vectors
and reducing the orthogonality defects of the basis matrices of lattices.
The KZ constant, defined by Schnorr in \cite{Sch87}, is a measure of the quality of KZ reduced matrices.
It can be used to bound the lengths of the column vectors of KZ reduced matrices from above \cite{LagLS90}, \cite{WenC18}.
In addition to this, the KZ constant has applications in bounding
the decoding radius and the proximity factors of KZ-aided successive interference cancellation
(SIC) decoders from below \cite{WenC18, Lin11, LuzSL13}.
Although the KZ constant is an important quantity, there is no formula for it.
Fortunately, it has several upper bounds \cite{Sch87}, \cite{HanS08},  \cite{WenC18}.
The first main aim of this paper is to improve the sharpest existing upper bound presented in \cite{WenC18}.

The Hermite constant can be used to quantify the length of the shortest nonzero vector of lattices.
Since estimating the length of the shortest vector in a lattice is a NP-hard problem \cite{Ajtai96},
this application of Hermite constant is of vital importance.
It also has applications in bounding the KZ constant from above \cite{Sch87}.
Furthermore, it can be used to derive lower bounds on
the decoding radius of the LLL-aided SIC decoders \cite{WenC18, LuzSL13},
and upper bounds on the orthogonality defect of KZ reduced matrices \cite{LagLS90}, \cite{WenC18}, \cite{LyuL17}.
Although the Hermite constant is important, its exact values are known for dimension $1\leq n\leq 8$ and $n=24$ only.
Thus, its upper bound for arbitrary integer $n$ is needed.
In the above applications, the Hermite constant's linear upper bounds play crucial roles.
Hence, in addition to the nonlinear upper bound \cite{Bli14}, several linear upper bounds on the Hermite constant
have been proposed in \cite{LagLS90}, \cite{NguV10}, \cite{Neu17}.
The second main aim of this paper is to improve the sharpest available linear upper bound provided in \cite{WenC18}.

%This paper investigate the linear upper bound on the Hermite constant and upper bound on the KZ constant.
%New upper bounds which are tighter than the sharpest existing upper bounds will be developed.

The reminder of the paper is organized as follows.
Sections \ref{s:HC} and \ref{s:KZ} develop a new  linear upper bound on the Hermite constant
and a new upper bound on the KZ constant, respectively.
Finally, this paper is summarized in Section \ref{s:sum}.

%In this paper, %let $\mathbb{R}^{m\times n}$ and $\mathbb{Z}^{m\times n}$ be the space of the $m\times n$
%%real matrices and integer matrices, respectively.
%boldface lowercase letters denote column vectors and boldface uppercase letters denote matrices.
%For a matrix $\A$, let $a_{ij}$ be its $(i,j)$ element and
%$\A_{i:j,k:\ell}$ be the submatrix containing elements with row indices from $i$ to $j$ and column indices from $k$ to $\ell$.
%Denote $\e_1=[1,0,\ldots, 0]^T$, whose dimension depends on the context.

{\it Notation.}
%Let $\mathbb{R}^n$ and $\mathbb{Z}^n$ be the spaces of the $n$-dimensional column real vectors and integer vectors, respectively.
Let $\mathbb{R}^{m\times n}$ and $\mathbb{Z}^{m\times n}$ be the spaces of the $m\times n$ real matrices and integer matrices, respectively.
Boldface lowercase letters denote column vectors and boldface uppercase letters denote matrices.
%e.g., $\y\in\mathbb{R}^n$ and $\A \in\mathbb{R}^{m\times n}$.
%For a vector $\x$, $\|\x\|_2$ denotes the $\ell^2$-norm of $\x$.
%Let $x_i$ be the element of $\x$ with index $i$.
For a matrix $\A$, we use $a_{ij}$ to denote its $(i,j)$ entry.
$\Gamma(n)$ denotes the Gamma function.
%$a_{ij}$ be the element at row  $i$ and column  $j$,
%and use $\A_{i:j,k:\ell}$ to denote the submatrix containing elements with row indices from $i$ to $j$ and column indices from $k$ to $\ell$.
%Let $\e_k$   denote   the $k$-th column of an identity matrix $\I$, whose dimension depends on the context.

\section{A sharper linear bound on the Hermite constant}
\label{s:HC}

This section develops a new linear upper bound on the Hermite constant.
that is sharper than \cite[Theorem 1]{WenC18} when $n\geq 109$.

We first introduce the definition of the Hermite constant.
Denote the set of $m\times n$ real matrices with full-column rank by $\mathbb{R}_n^{m\times n}$.
The Hermite constant $\gamma_n$ is defined as
\[
\gamma_n=\sup_{\A\in\mathbb{R}_n^{m\times n}}\frac{(\lambda (\A))^2}{(\det(\A^T\A))^{1/n}},
\]
where $\lambda(\A)$ represents the length of a shortest nonzero vector of ${\cal L}(\A)$, i.e.,
$$
\lambda(\A)=\min_{\x\in \Zn\backslash \{\0\}} \|\A\x\|_2.
$$

Although the Hermite constant is a vital important constant of lattices, the values of $\gamma_n$
are known only for $n=1,\ldots,8$ \cite{Mar13} and $n=24$ \cite{CohA04}
(see also \cite[Table 1]{WenC18}).
Fortunately, there are some upper bounds on $\gamma_n$ for any $n$ in the literature and the sharpest one is
\beq
\label{e:blichfeldt}
\gamma_n \leq \frac{2}{\pi} (\Gamma(2+n/2))^{2/n},
\eeq
given by Blichfeldt \cite{Bli14}.

As explained in Section \ref{s:introduction}, linear upper bounds on $\gamma_n$ are very useful.
There are several linear upper bounds: $\gamma_n \leq \frac{2}{3}n$ (for $n\geq 2$) \cite{LagLS90};
$\gamma_n \leq 1 + \frac{n}{4}$ (for $n\geq 1$) \cite[p.35]{NguV10} and
$\gamma_n \leq \frac{n+6}{7} $ (for $n\geq 2$) \cite{Neu17}.
The most recent linear upper bound on $\gamma_n$ is
\beq
\label{e:gammanWC}
\gamma_n < \frac{n}{8}+\frac{6}{5}, \,\; n\geq 1,
\eeq
given in \cite[Theorem 1]{WenC18}.

The following theorem gives a new linear upper bound on $\gamma_n$,
which is sharper than \eqref{e:gammanWC} when $n\geq109$.

\begin{theorem}
\label{t:gamman}
For $n\geq 1$,
\beq
\label{e:gamman}
\gamma_n < \frac{n}{8.5}+2.
\eeq
\end{theorem}

\begin{proof}
By \eqref{e:blichfeldt}, to show \eqref{e:gamman}, it suffices to show
\[
\left(\Gamma\left(2+\frac{n}{2}\right)\right)^{2/n} <  \frac{\pi(n+17)}{17},
\]
which is equivalent to
\beq
\label{e:lbd}
\Gamma\left(2+\frac{n}{2}\right) < \left(\frac{\pi(n+17)}{17}\right)^{n/2}.
\eeq
Then, to show \eqref{e:gamman}, it is equivalent to show that
\[
\phi(t):=\frac{\left[\frac{\pi}{8.5}(t+8.5)\right]^{t}}{\Gamma\left(2+t\right)} > 1
\]
for $t= 0.5, 1, 1.5, 2, 2.5, \ldots$.

By some direct calculations, one can show that
\[
\phi(t)>1, \mbox{ for }t= 0.5, 1.5, 2, 2.5, \ldots, 310.
\]
Thus, to show \eqref{e:gamman}, we only need to show that $\phi(t)$ or $\bar{\phi}(t):=\ln(\phi(t))$
is monotonically  increasing when $t\geq 310$.
% Equivalently, we need to show that $\bar{\phi}(t)$ is monotonically  increasing when $t\geq 60$.
%where
%\begin{align*}
%\bar{\phi}(t):=\ln(\phi(t))=\ &  t \ln \left[\frac{\pi}{8.5}(t+8.5)\right]-\ln \Gamma\left(2+t\right).
%\end{align*}

By some direct calculations, we have
\begin{align*}
\bar{\phi}'(t)= \ln \left[\frac{\pi}{8.5}(t+8.5)\right] +\frac{t}{t+8.5}-\psi(t+2),
\end{align*}
where $\psi(t+2)$ is the digamma function, i.e., $\psi(t+2)=\Gamma'(t+2)/\Gamma(t+2)$.
Then, to show \eqref{e:gamman}, we only need to show that $\bar{\phi}'(t)\geq 0$ when $t\geq 310$.
To achieve this, we use the following inequality from \cite[eq. (1.7) in Lemma 1.7]{Bat08}:
\beq
\label{e:digammaud}
\psi(t+2)\leq \ln (t+e^{1-\gamma}), \mbox{ for } t\geq 0,
\eeq
where $\gamma=\lim_{n\rightarrow \infty} (-\ln n + \sum_{k=1}^n 1/k)$, which is referred to as Euler's constant.
 Then, from the expression of $\bar{\phi}'(t)$ given before, we have
 $$
 \bar{\phi}'(t)\geq \rho(t),
$$
 where
\begin{align*}
\rho(t):&=\ln \left[\frac{\pi(t+8.5)}{8.5}\right] +\frac{t}{t+8.5}-\ln (t+e^{1-\gamma})\\
&=\ln (t+8.5)-\frac{8.5}{t+8.5}-\ln (t+e^{1-\gamma}) +\ln\frac{\pi e}{8.5}.
\end{align*}
Since
\begin{align*}
\rho'(t) & = \frac{1}{t+8.5} +\frac{8.5}{(t+8.5)^2} - \frac{1}{t+e^{1-\gamma}} \\
&=  \frac{(t+8.5)(t+e^{1-\gamma})+8.5(t+e^{1-\gamma})}{(t+8.5)^2(t+e^{1-\gamma})} \\
&\quad-  \frac{(t+8.5)^2}{(t+8.5)^2(t+e^{1-\gamma})} \\
&=\frac{e^{1-\gamma}t-(72.25-17e^{1-\gamma})}{(t+8.5)^2(t+e^{1-\gamma})},
\end{align*}
and $\gamma<0.58$ \cite{Leh75}, $\rho'(t)\geq0$ when $t>31$ as
\begin{align*}
&e^{1-\gamma}t-(72.25-17e^{1-\gamma})\\
>&31\times e^{1-\gamma}-(72.25-17e^{1-\gamma})>0.
\end{align*}
Thus, for $t\geq 310$, we have
\[
\bar{\phi}'(t)\geq \rho(t)\geq \rho(310) > 0.0000796>0,
\]
where the third inequality follows form the fact that $\gamma>0.57$ \cite{Leh75}.
\end{proof}

By some simple calculations, one can easily see that the upper bound \eqref{e:gamman} is sharper than
the upper bound \eqref{e:gammanWC} when $n\geq 109$.
When $n\leq 108$,  \eqref{e:gammanWC} is sharper than \eqref{e:gamman}, but their difference is small.
By the Stirling's approximation for Gamma function, the right-hand side of \eqref{e:blichfeldt}
is  asymptotically $\frac{n}{\pi e}\approx\frac{n}{8.54} $. Thus, the linear bound given by \eqref{e:gamman}
is very close to the nonlinear upper bound given by \eqref{e:blichfeldt}.
To clearly show the improvement of \eqref{e:gamman} over \eqref{e:gammanWC} and how close \eqref{e:gamman} is
to \eqref{e:blichfeldt},  in Figure \ref{f:HC} we plot the ratios of the two bounds to
Blichfeldt's bound given by \eqref{e:blichfeldt}:
$$
\mbox{Ratio 1} = \frac{\frac{n}{8}+\frac{6}{5}}{\frac{2}{\pi} (\Gamma(2+n/2))^{2/n}},  \;
\mbox{Ratio 2} = \frac{\frac{n}{8.5}+2}{\frac{2}{\pi} (\Gamma(2+n/2))^{2/n}}.
$$
\begin{figure}[!htbp]
\centering
\includegraphics[width=3.4 in]{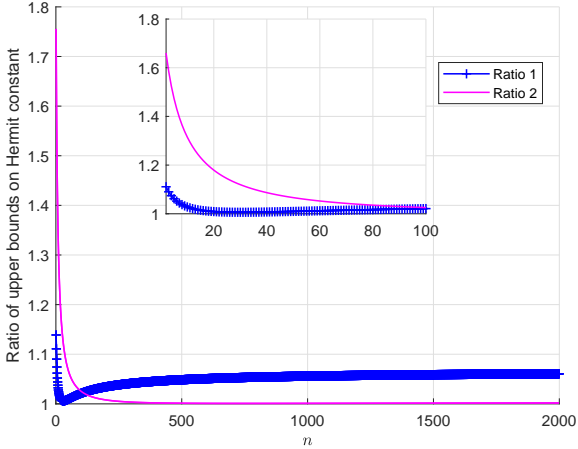}
\caption{The ratio of the bounds given by \eqref{e:gamman} and  \eqref{e:gammanWC} to
Blichfeldt's bound in \eqref{e:blichfeldt}  versus $n$}
\label{f:HC}
\end{figure}

From Figure \ref{f:HC}, one can see that the upper bound given by \eqref{e:gamman} is very close to the nonlinear
upper bound given by \eqref{e:blichfeldt}, and \eqref{e:gamman} improves \eqref{e:gammanWC} for $n\geq 109$.

In the following, we give some remarks.
\begin{remark}
\label{r:prof}
The approach used by the proof for  \eqref{e:gamman}
is different from that for  \eqref{e:gammanWC}  used in \cite{WenC18}.
To show \eqref{e:gammanWC}, it suffices to show (cf.\ \eqref{e:lbd})
\beq \label{e:gubdWC}
\Gamma\left(2+\frac{n}{2}\right) < \left(\frac{\pi(n+9.6)}{16}\right)^{n/2}.
\eeq
The proof for \eqref{e:gubdWC} first gives an upper bound on $\Gamma\left(2+\frac{n}{2}\right)$
and then shows the right-hand side of
\eqref{e:gubdWC} is larger than this upper bound,
while  the proof for  \eqref{e:lbd} here shows $\phi(t)$ is a monotonically  increasing function
by using an upper bound on the digamma function (see \eqref{e:digammaud}).
\end{remark}

\begin{remark}
\label{r:LLLdr}
The improved linear upper bound \eqref{e:gamman} on $\gamma_n$ can be used to improve
the lower bound on the decoding radius of the LLL-aided  SIC decoder
that was given in \cite{WenC18},
which is an improvement of the one given  in  \cite[Lemma 1]{LuzSL13}.
%Since the improved lower bound on the decoding radius can be easily obtained
Since the derivation for the new  lower bound on the decoding radius is straightforward
by following the proof of \cite[Lemma 1]{LuzSL13} and using \eqref{e:gamman},
we do not provide details.
\end{remark}

\begin{remark}
\label{r:LLLOD}
The improved linear upper bound \eqref{e:gamman} on $\gamma_n$ can be used to improve
the upper bound on the orthogonality defect of KZ reduced matrices
that was presented in \cite[Theorem 4]{WenC18}.
Note that the orthogonality defect of a matrix is a good measure of the orthogonality
of the matrix and hence it is often used in characterizing the quality of
a LLL or KZ reduced matrix.
\end{remark}

\section{A sharper bound on the KZ constant}
\label{s:KZ}

In this section, we develop an upper bound on the KZ constant that is sharper than that given by \cite[Theorem 2]{WenC18}.

We first briefly introduce the definition of the KZ reduction.
Suppose that $\A$ in \eqref{e:latticeA} has the following thin QR factorization (see, e.g., \cite[Chap.\ 5]{GolV13}):
\beq
\label{e:QR}
\A= \Q \R,
\eeq
where $\Q\in \Rmbn$ has orthonormal columns and $\R\in \Rnbn$ is nonsingular upper triangular,
and they are respectively referred to as $\A$'s Q-factor and R-factor.
If $\R$ in \eqref{e:QR} satisfies:
\begin{align}
%\label{e:criteria1}
&|r_{ij}|\leq\frac{1}{2} |r_{ii}|, \ \ 1\leq i \leq j-1 \leq n-1,  \\
&| r_{ii}| =\min_{\x\,\in\,\mathbb{Z}^{n-i+1}\backslash \{\0\}}\|\R_{i:n,i:n}\x\|_2, \ \
1\leq i \leq n,
%\label{e:criteria2}
\end{align}
then $\A$ and $\R$ are said to be  KZ reduced.
Given $\A\in \mathbb{R}_n^{m\times n}$, the KZ reduction is the process of finding a unimodular matrix
$\Z\in \mathbb{Z}^{n\times n}$ such that $\A\Z$ is  KZ reduced.

%To quantify the quality of KZ (and block KZ) reduced matrices, Schnorr defined KZ constant $\alpha_n$ in \cite{Sch87}.
Let $\mathcal{B}_{KZ}$ denote the set of all $m\times n$ KZ reduced matrices with full-column rank.
The KZ constant   is  defined as \cite{Sch87}
\beq
\label{e:KZconstant}
\alpha_n=\sup_{\A\in \mathcal{B}_{KZ}} \frac{(\lambda (\A))^2}{r_{nn}^2},
\eeq
where $\lambda (\A)$ denotes the length of the shortest nonzero vector of $\mathcal{L}(\A)$,
and $r_{nn}$ is the last diagonal entry of the R-factor $\R$ of $\A$ (see \eqref{e:QR}).

As explained in Section \ref{s:introduction}, the KZ constant is an important quantity  for characterizing
some properties of KZ reduced matrices.  However, its exact value is unknown.
Hence, it is useful to find a good upper bound on it.
Schnorr in \cite[Corollary 2.5]{Sch87} proved that
\[
\alpha_n\leq n^{1+\ln n}, \ \ \mbox{for } n\geq 1;
\]
 Hanrot and Stehl{\'e} in \cite[Theorem 4]{HanS08} showed that
\[
\alpha_n\leq n\prod_{k=2}^nk^{1/(k-1)} \leq n^{\frac{\ln n}{2}+\bigO(1)}, \ \ \mbox{for } n\geq 2;
\]
Based on the exact value  of $\gamma_n$  for $1\leq n\leq 8$ and the upper bound
on $\gamma_n$  in \eqref{e:gammanWC} for $n\geq 9$, Wen and Chang in \cite[Theorem 2]{WenC18}
 showed that
\beq
\label{e:KZconstantUB1}
\alpha_n\leq 7 \left( \frac{1}{8}n +\frac{6}{5}\right)  \left(\frac{n-1}{8}\right)^{\frac{1}{2}\ln((n-1)/8)}, \ \ \mbox{for } n\geq 9.
\eeq

%In the following, we will develop a new upper bound on $\alpha_n$ for $n\geq 9$ and then show
%that it is sharper than that given by \ref{e:KZconstantUB1}.

In the following theorem we provide a new upper bound on $\alpha_n$ for $n\geq109$,
which is sharper than that in \eqref{e:KZconstantUB1} for $n\geq 111$.
The new bound on $\alpha_n$ is based on the new upper bound on
the Hermite constant $\gamma_n$ \eqref{e:gamman},
which is sharper than that in \eqref{e:gammanWC}  for $n\geq 109$.

\begin{theorem}
\label{t:KZconstantUB}
The KZ constant $\alpha_n$ satisfies
\beq
\label{e:KZconstantUB2}
\alpha_n\leq  8.1 \left( \frac{n}{8.5} +2\right)  \left(\frac{2n-1}{17}\right)^{\frac{1}{2}\ln((2n-1)/17)}, \ \ \mbox{for } n\geq 109.
\eeq
\end{theorem}

%Since the proof of Theorem \ref{t:KZconstantUB} is complicated, we provide it in Appendix \ref{s:proofKZC}.

To prove Theorem \ref{t:KZconstantUB}, we need to introduce two lemmas.
The first one is from \cite[Lemma 2]{WenC18}.

\begin{lemma}\label{l:integralbd}
For $a>b>0$ and $c>0$
\beq
\label{e:integralbd}
 \int_a^b \frac{\ln( 1+c/t)}{t} d t \leq
 \frac{9}{8} \ln \frac{b(3a+2c)}{a(3b+2c)} + \frac{c(b-a)}{4ab}.
\eeq
\end{lemma}

The second lemma which is needed for proving Theorem \ref{t:KZconstantUB} is as follows:
\begin{lemma}\label{l:integralbd2}
Suppose that $f(t)$ satisfies $f''(t)\geq 0$ for $t\in [a,b]$.
Then
\beq
\label{e:integralbd2}
(b-a) f\left(\frac{a+b}{2}\right)\leq \int_a^b f(s) d s.
\eeq
\end{lemma}

\begin{proof}
%Let
%\[
%F(t)=\int_a^t f(s) d s-(t-a)f(\frac{a+t}{2}), \, \;t\geq a,
%\]
%then $F(a)=0$. Hence, to show \eqref{e:integralbd2}, it suffices to show that $F'(t)\geq 0$ for $t\geq a$.
%
%By some simple calculations, we have
%\[
%F'(t)=f(t)-f(\frac{a+t}{2})-\frac{t-a}{2}f'(\frac{a+t}{2}).
%\]
%By the mean value theorem, there exist $\eta\in[\frac{a+t}{2},t]$ such that
%\[
%f(t)-f(\frac{a+t}{2})=\frac{t-a}{2}f'(\eta).
%\]
%Since $\eta\in[\frac{a+t}{2},t]$, $f''(t)\geq 0$, we have $f'(\eta)\geq f'(\frac{a+t}{2})$,
%hence $F'(t)\geq 0$. This completes the proof.
The left hand side of \eqref{e:integralbd2} is referred to as the midpoint rule
for approximating  the integral on the right hand side in numerical analysis.
It is well known that
\beq \label{e:midpoint}
 \int_a^b f(s) d s - (b-a) f\left(\frac{a+b}{2}\right)
= \frac{1}{24}(b-a)^3 f''(z)
\eeq
for some $z \in (a,b)$.
This formula can be easily proved as follows.
By Taylor's theorem,
\begin{align*}
f(s) = & f\left(\frac{a+b}{2}\right) + f'\left(\frac{a+b}{2}\right) \Big(s- \frac{a+b}{2}\Big) \\
& + \frac{1}{2} f''(\zeta(s)) \Big(s- \frac{a+b}{2}\Big)^2,
\end{align*}
where $\zeta(s)$ depends on $s\in (a, b)$.
Integrating both sides of the above equality over $[a,b]$
and using the Mean-Value-Theorem for Integrals
immediately lead to \eqref{e:midpoint}.
Then using the condition that $f''(t)\geq 0$ for $t\in [a,b]$,
we obtain \eqref{e:integralbd2}.
\end{proof}

In the following, we give a proof for Theorem \ref{t:KZconstantUB}
by following the proof of \cite[Theorem 2]{WenC18}.
\begin{proof}
According to the proof of \cite[Cor. 2.5]{Sch87},
\beq
\label{e:KZUB11}
\alpha_n\leq \gamma_n\prod_{k=2}^n\gamma_k^{1/(k-1)}.
\eeq
By \cite[(53)]{WenC18}, we have
\beq
\label{e:KZUB12}
\prod_{k=2}^8\gamma_k^{1/(k-1)}= 2^{\frac{827}{420}} 3^{-\frac{8}{15}}.
\eeq
By \eqref{e:gammanWC}, we obtain
\beq
\label{e:KZUB13}
\prod_{k=9}^{108}\gamma_k^{1/(k-1)}\leq \prod_{k=9}^{108}\left(\frac{k}{8}+\frac{6}{5}\right)^{1/(k-1)}< 79.06.
\eeq

In the following, we use Theorem \ref{t:gamman} to bound $\prod_{k=109}^n\gamma_k^{1/(k-1)}$ from above.
By Theorem \ref{t:gamman}, we obtain
\begin{align}
% \label{e:KZconstantUB6}
 \prod_{k=109}^n\gamma_k^{1/(k-1)}
\leq &  \prod_{k=109}^n\left(\frac{k}{8.5} + 2\right)^{1/(k-1)} \nonumber  \\
= &  \prod_{k=108}^{n-1}\left(\frac{ k+18}{8.5} \right)^{1/k}    \nonumber  \\
= & \exp\left[\sum_{k=108}^{n-1}\frac{1}{k}\ln \left(\frac{ k+18}{8.5} \right)\right]\nonumber\\
\overset{(a)}{\leq}&\exp\left(\sum_{k=108}^{n-1}\int_{k-0.5}^{k+0.5}\frac{1}{t} \ln  \left(\frac{ t+18}{8.5} \right)dt\right)\nonumber\\
%=&\exp\left(\int_{7}^{n-1} \frac{1}{t}\ln\left(\frac{1}{8}\left(t+\frac{53}{5}\right)\right)dt\right)\nonumber\\
=&\exp\!\left(\int_{107.5}^{n-0.5}\frac{1}{t} \ln\left(\frac{t+18}{t}\frac{t}{8.5}\right)dt\right)\nonumber\\
 = &\exp\!\left(\int_{107.5}^{n-0.5} \frac{1}{t} \ln \left(1 +\frac{18}{t}\right) dt \! \right )\nonumber\\
 &\times \exp\! \left(\int_{107.5}^{n-0.5}\frac{\ln(t/8.5)}{t}dt \! \right),   \label{e:kz2terms}
\end{align}
where (a) follows from Lemma \ref{l:integralbd2} with $a=k-0.5, b=k+0.5$ and the fact that
 for $t\geq107.5$,
$\omega(t):=\frac{1}{t} \ln \left(\frac{ t+18}{8.5} \right)$ satisfies
\begin{align*}
\omega''(t)&\!=\!\frac{1}{t^3(t+18)^2}\left(\! 2(t+18)^2\ln \left(\frac{t+18}{8.5}\right) -(3t^2+36t)\! \right)\\
&\geq \frac{1}{t^3(t+18)^2}\left(\! 2(t+18)^2\cdot\ln \frac{125.5}{8.5} -(3t^2+36t)\! \right)\\
&\geq \frac{1}{t^3(t+18)^2}\left(2(t+18)^2\cdot2 -(3t^2+36t)\right)\geq 0.
\end{align*}

Now we bound the two factors on the right-hand side of \eqref{e:kz2terms} from above.
By Lemma \ref{l:integralbd}, we obtain
\begin{align}
& \exp\left( \int_{107.5}^{n-0.5}\frac{1}{t} \ln  \left(1 +\frac{18}{t}\right) dt \right)  \nonumber  \\
\leq \ &  \exp\left( \frac{9}{8} \ln \frac{358.5(n-0.5)}{107.5(3(n-0.5)+36)} +\frac{18(n-108)}{430(n-0.5)}  \right)  \nonumber  \\
\leq \ &  \exp\left( \frac{9}{8} \ln \frac{358.5(n-0.5)}{107.5\times3(n-0.5)} +\frac{18(n-108)}{430(n-108)}  \right)  \nonumber  \\
= \ &   \left(\frac{119.5}{107.5}\right)^{9/8} \exp\left( \frac{9}{215}\right). \label{e:1stub}
\end{align}

By a direct calculation, we have
\begin{align}
 & \exp\left(\int_{107.5}^{n-0.5}\frac{\ln(t/8.5)}{t}dt\right) \nonumber\\
=\ & \exp\left(\frac{\ln^2((n-0.5)/8.5)}{2}-\frac{\ln^2(107.5/8.5)}{2}\right)   \nonumber\\
 =\ & \left(\frac{n-0.5}{8.5}\right)^{\frac{1}{2}\ln((n-0.5)/8.5)}
 \left(\frac{8.5}{107.5}\right)^{\frac{1}{2}\ln(107.5/8.5)} \label{e:2ndub}
% <\ & \left(\frac{n-1}{8}\right)^{\ln((n-1)/8)/2}, \nonumber\\
%<\ & \exp\left(\frac{\ln^2((n-1)/8)}{2}\right)\nonumber\\
%=\ & \left(\frac{n-1}{8}\right)^{\ln((n-1)/8)/2},
\end{align}

Then combining  \eqref{e:KZUB11}-\eqref{e:2ndub} and \eqref{e:gamman}, we obtain that
for $n\geq 109$
\begin{align*}
 \alpha_n
\leq \ & 79.06\times 2^{\frac{827}{420}} 3^{-\frac{8}{15}}
       \left(\frac{119.5}{107.5}\right)^{9/8} \exp\left( \frac{9}{215}\right) \\ &  \times       \left(\frac{8.5}{107.5}\right)^{\frac{1}{2}\ln\frac{107.5}{8.5}}
\left( \frac{n}{8.5} +2\right)  \left(\frac{n-0.5}{8.5}\right)^{\frac{1}{2}\ln(\frac{n-0.5}{8.5})} \\
< \ & (8.0911 \cdots)   \left( \frac{n}{8.5} +2\right)  \left(\frac{n-0.5}{8.5}\right)^{\frac{1}{2}\ln(\frac{n-0.5}{8.5})} \\
< \ & 8.1 \left( \frac{n}{8.5} +2\right)  \left(\frac{2n-1}{17}\right)^{\frac{1}{2}\ln((2n-1)/17)}.
\end{align*}
\end{proof}

\begin{remark}
Note that although the proof of Theorem \ref{t:KZconstantUB} is similar to the proof of
\cite[Theorem 2]{WenC18}, there is some difference between them.
The main difference between them is (a) in \eqref{e:kz2terms}.
Here, we use Lemma \ref{l:integralbd2} to build (a),
while the proof of \cite[Theorem 2]{WenC18} uses the decreasing property of the integrand
to get the inequality.
\end{remark}

To clearly see the improvement of \eqref{e:KZconstantUB2} over \eqref{e:KZconstantUB1},
we draw the ratio of the right-hand side of \eqref{e:KZconstantUB2} to that of \eqref{e:KZconstantUB1}
for $n=111:1:1000$ in Figure \ref{f:KZ}.
The figure shows that \eqref{e:KZconstantUB2} significantly outperforms  \eqref{e:KZconstantUB1},
and the improvement becomes more significant as $n$ gets larger.

\begin{figure}[!htbp]
\centering
\includegraphics[width=3.4 in]{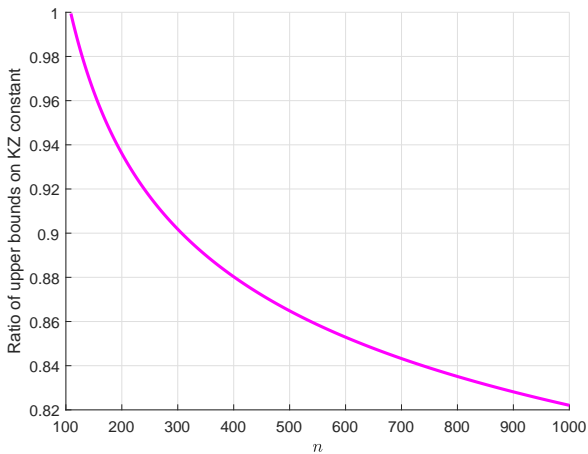}
\caption{The ratio of the bound given by \eqref{e:KZconstantUB2} to
the bound given by \eqref{e:KZconstantUB1}  versus $n$}
\label{f:KZ}
\end{figure}

In the following we give remarks about some applications of Theorem \ref{t:KZconstantUB}.

\begin{remark}
\label{r:KZpf}
As in \cite[Remark 2]{WenC18}, we can use the improved upper bound \eqref{e:KZconstantUB2} on $\alpha_n$
to derive upper bounds on the proximity factors of the KZ-aided SIC decoder and these new  bounds are sharper than those given in \cite[Remark 2]{WenC18}.
%which are much tight than those given by \cite[eq.s (41) and (45)]{Lin11}.
Since the derivations are straightforward, we omit its details.
\end{remark}

\begin{remark}
\label{r:KZdr}
We can use  \eqref{e:KZconstantUB2} and follow the proof of \cite[Lemma 1]{LuzSL13} to derive a lower bound on  the decoding radius of the KZ-aided SIC decoder, which
is tighter than that given in \cite[Remark 3]{WenC18} when $n\geq 111$.
\end{remark}

\begin{remark}
\label{r:KZod}
By following the proof of \cite[Theorem 3]{WenC18} and using  \eqref{e:KZconstantUB2},
we can also develop new upper bounds on the lengths of the KZ reduced matrices,
which are tighter than those given in \cite[Theorem 3]{WenC18} when $n\geq 111$.
\end{remark}

\section{Summary} \label{s:sum}
The KZ reduction is one of the most popular lattice reduction methods and has many important applications.
The Hermite constant is a basic constant of lattice.
In this paper, we first developed a new linear upper bound on the Hermite constant
and then utilized the bound to develop a new upper bound on the KZ constant.
These bounds are sharper than those developed in \cite{WenC18}.
Some  applications of the new sharper bounds on the Hermite and KZ constants were also discussed.

\appendices

\bibliographystyle{IEEEtran}
\bibliography{ref}

\end{document}